\pdfoutput=1
\RequirePackage{ifpdf}
\ifpdf 
\documentclass[pdftex]{sigma}
\else
\documentclass{sigma}
\fi

\numberwithin{equation}{section}

\newtheorem{Theorem}{Theorem}[section]
\newtheorem{Corollary}[Theorem]{Corollary}
\newtheorem{Lemma}[Theorem]{Lemma}
\newtheorem{Proposition}[Theorem]{Proposition}
 { \theoremstyle{definition}
\newtheorem{Definition}[Theorem]{Definition}
\newtheorem{Example}[Theorem]{Example}
\newtheorem{Remark}[Theorem]{Remark} }

\newcommand{\B}{\mathcal{B} }
\newcommand{\bC}{{\mathbb C}}
\newcommand{\bZ}{{\mathbb Z}}
\newcommand{\bN}{{\mathbb N}}

\begin{document}

\newcommand{\arXivNumber}{1805.06971}

\renewcommand{\thefootnote}{}

\renewcommand{\PaperNumber}{065}

\FirstPageHeading

\ShortArticleName{Multiparameter Schur $Q$-Functions Are Solutions of the BKP Hierarchy}

\ArticleName{Multiparameter Schur $\boldsymbol{Q}$-Functions\\ Are Solutions of the BKP Hierarchy\footnote{This paper is a~contribution to the Special Issue on Representation Theory and Integrable Systems in honor of Vitaly Tarasov on the 60th birthday and Alexander Varchenko on the 70th birthday. The full collection is available at \href{https://www.emis.de/journals/SIGMA/Tarasov-Varchenko.html}{https://www.emis.de/journals/SIGMA/Tarasov-Varchenko.html}}}

\Author{Natasha ROZHKOVSKAYA}

\AuthorNameForHeading{N.~Rozhkovskaya}

\Address{Department of Mathematics, Kansas State University, Manhattan, KS 66502, USA}
\Email{\href{mailto:rozhkovs@math.ksu.edu}{rozhkovs@math.ksu.edu}}
\URLaddress{\url{http://www.math.ksu.edu/~rozhkovs/}}

\ArticleDates{Received May 20, 2019, in final form August 23, 2019; Published online August 28, 2019}

\Abstract{We prove that multiparameter Schur $Q$-functions, which include as specializations factorial Schur $Q$-functions and classical Schur $Q$-functions, provide solutions of the BKP hierarchy.}


\Classification{05E05; 17B65; 17B69; 11C20}

\renewcommand{\thefootnote}{\arabic{footnote}}
\setcounter{footnote}{0}

\section{Introduction}

Integrable systems of the KP (Kadomtsev--Petviashvilli) type hierarchy of partial differential equations, which corresponds to infinite-dimensional Lie algebra of type A, and of its type B variant, the BKP hierarchy, have as solutions renowned families of symmetric functions~--
Schur polynomials in the KP case, and Schur $Q$-polynomials in the BKP case \cite{DKTIV,DKTA,DKTII, JM,JMbook1,Sato,You1,You2}, etc.
In this note we show that multiparameter Schur $Q$-functions also provide solutions of the BKP hierarchy.

Multiparameter Schur $Q$-functions $Q_\lambda ^{(a)}$ were introduced and studied combinatorially in \cite{Iv1}. These symmetric functions are interpolation analogues of the classical Schur $Q$-functions depending on a sequence of complex valued parameters
$a=(a_0, a_1,\dots)$. The definition of multiparameter Schur $Q$-functions is reproduced in (\ref{defQa}). {\it Classical} Schur $Q$-functions correspond to $a=(0,0,0,\dots)$, and with the evaluation $a=(0, 1,2, 3,\dots)$ the multiparameter Schur $Q$-functions are called {\it factorial} Schur $Q$-functions. These families of symmetric functions proved to be useful in study of a number of questions of representation theory and algebraic geometry. Here are a few examples.

The authors of \cite{ASS, Naz,Serg} described Capelli polynomials of the queer Lie superalgebra which form a distinguished family of super-polynomial differential operators indexed by strict partitions acting on an associative superalgebra. The eigenvalues of these Capelli polynomials are expressed through the factorial Schur $Q$-functions.

In \cite{Ikeda, Ikeda2} the equivariant cohomology of a Lagrangian Grassmannian of a symplectic or orthogonal types is studied. The restrictions of Schubert classes to the set of points fixed under the action of a maximal torus of the symplectic group are calculated in terms of factorial symmetric functions. Further in \cite{Leon1} factorial Schur $Q$-functions are used to write generators and relations for the equivariant quantum cohomology rings of the maximal isotropic Grassmannians of types~B,~C and~D.

In \cite{Henry} the center of the twisted version of Khovanov's Heisenberg category is identified with the algebra generated by classical Schur $Q$-functions (denoted as $\B_{\rm odd}$ in the exposition below). Factorial Schur $Q$-functions are described as closed diagrams of this category.

The goal of this note is to show that multiparameter Schur $Q$-functions $Q_\lambda ^{(a)}$ are solutions of the BKP hierarchy. The origin for this phenomena lies in the fact proved in \cite{Kor} that generating functions of multiparameter Schur $Q$-functions and of classical Schur $Q$-functions coincide.

While the BKP hierarchy is described in a wide range of literature on integrable systems and solitons, for the completeness of exposition and for the convenience of the reader we formulate the whole setting of the BKP hierarchy in terms of generating functions of symmetric functions with the neutral fermions bilinear identity~(\ref{binBKP}) as a starting point. We avoid to use any other facts than the well-known properties of symmetric functions that can be found in the classical monograph~\cite{Md}, and through the text we provide the references to the corresponding chapters and examples of that monograph.

It is worth to mention that formulation of the KP and the BKP integrable systems solely in terms of symmetric functions can be found, e.g., in~\cite{JY}. The authors of \cite{JY} start with the bilinear identities in integral form, then, using the Cauchy type orthogonality properties of symmetric functions (cf.~\cite[Chapter~III, equation~(8.13)]{Md}), they arrive at Plucker type relations, and the later ones are transformed into the collection of partial differential equations of Hirota derivatives that constitute the hierarchy. As it is mentioned above, our route is traced differently employing the properties of generating functions of complete, elementary symmetric functions and power sums. We obtain differential equations of the hierarchy in Hirota form as coefficients of Taylor expansions. One of the advantages of this approach is that it directly addresses the corresponding vertex operators actions, since the later ones are also `generating functions' (formal distributions).

The paper is organized as follows. In Section~\ref{section2} we recall some facts about complete, elementary symmetric functions, power sums and classical Schur $Q$-functions. In Section~\ref{section3} we describe the action of neutral fermions on the space generated by classical Schur $Q$-functions. In Section~\ref{section4} we review properties of generating functions for multiplication operators and corresponding adjoint operators and deduce vertex operator form of the formal distribution of neutral fermions. In Section~\ref{section5} we review all the steps of recovering the BKP hierarchy of partial differential equations in Hirota form from the neutral fermions bilinear identity. In Section~\ref{section6} we make simple observation that immediately shows that classical Schur $Q$-functions are solutions of the BKP hierarchy (which recovers the result of~\cite{You1}). In Section~\ref{multi_sec} we introduce multiparameter Schur $Q$-functions, and using the observation of Section~\ref{section6}, we show that $Q_\lambda ^{(a)}$ are also solutions of the BKP hierarchy.

\section[Schur $Q$-functions]{Schur $\boldsymbol{Q}$-functions}\label{section2}

Let $\B$ be the ring of symmetric functions in variables $(x_1, x_2, \dots)$. Consider the families of the following symmetric functions:
\begin{alignat*}{3}
&\text{elementary symmetric functions}\qquad && \left\{e_k=\sum_{ i_1<\dots < i_k} x_{i_1}\cdots x_{i_k}\,| \, {k=0,1,\dots}\right\}, & \\
&\text{complete symmetric functions} \qquad && \left\{h_k=\sum_{ i_1\le \dots \le i_k} x_{i_1}\cdots x_{i_k}\,|\, {k=0,1,\dots}\right\}, \\
&\text{symmetric power sums} \qquad && \left\{p_k=\sum x_i^k \,|\,{k=0,1,\dots }\right\}.
\end{alignat*}
We set $e_k=h_k=0$ for $k<0$. It is well-known \cite[Chapter I.2]{Md}, that each of these families generate $\B$ as a polynomial ring:
\begin{gather*}
\B=\bC[p_1, p_2, p_3,\dots]= \bC[e_1, e_2, e_3,\dots]= \bC[h_1, h_2, h_3,\dots].
\end{gather*}

Combine the families $h_k$, $e_k$, $p_k $ into generating functions
\begin{gather*}
H(u)=\sum_{k\ge 0} \frac{h_k}{u^k},\qquad E(u)=\sum_{k\ge 0} \frac{e_k}{u^k}, \qquad P(u)=\sum_{k\ge 1} \frac{p_k}{u^k}.
\end{gather*}

The following facts are well-known \cite[Chapter~I.2]{Md}.
\begin{Lemma}\label{Lemma1}
\begin{gather*}
H(u)=\prod_i \frac{1}{1-x_i/u},\qquad E(u)=\prod_i {1+x_i/u},\qquad E(-u)H(u)=1,\\
H(u)=\exp\left(\sum_{n\ge1}\frac{1}{n} \frac{p_n}{u^{n}}\right),\qquad
E(u)=\exp\left(\sum_{n\ge1}\frac{(-1)^{n-1}}{n} \frac{p_{n}}{u^n}\right).
\end{gather*}
\end{Lemma}

We introduce one more family of symmetric functions $\{Q_k=Q_k(x_1, x_2,\dots)\}$ with $(k=0,1,\dots)$ as the coefficients of the generating function
\begin{gather}\label{Qk1}
Q(u)=\sum_{k\ge 0} \frac{Q_k}{u^k}=\prod_i \frac{u+x_i}{u-x_i}.
\end{gather}
From Lemma \ref{Lemma1} and (\ref{Qk1}) we immediately get relations of the next lemma.
 \begin{Lemma}\label{Lemma2}
\begin{gather*}
Q(u) = E(u) H(u)= R(u)^2, \qquad R(u)=\exp\left(\sum_{n\in \bN_{\rm odd}}\frac{p_{n} }{nu^{n}}\right),
\end{gather*}
where $\bN_{\rm odd}=\{1,3,5,\dots\}$.
\end{Lemma}
Note that $Q(u) Q(-u)=1$, which implies that $Q_r$ with even $r$ can be expressed algebraically through $Q_r$ with odd $r$:
\begin{gather*}
Q_{2m}=\sum_{r=1}^{m-1}(-1)^{r-1} Q_rQ_{2m-r} +\frac{1}{2} (-1)^{m-1}Q^2_m.
\end{gather*}
More generally, Schur $Q$-functions $Q_\lambda$ labeled by strict partitions are defined as a specialization of Hall--Littlewood polynomials \cite[Chapter~III.2]{Md}.
\begin{Definition}Let $\lambda=(\lambda_1>\lambda_2>\dots> \lambda_l)$ be a strict partition. Let $l\le N$. Schur $Q$-polynomial $Q_\lambda (x_1,\dots, x_N)$ is the symmetric polynomial in variables $x_i$'s defined by the formula
\begin{gather}\label{defQ}
Q_\lambda (x_1,\dots, x_N)=\frac{2^l}{(N-l)!}
\sum_{\sigma\in S_N} \prod_{i=1}^{l}x_{\sigma(i)}^{\lambda_i} \prod_{i<j}\frac{x_{\sigma(i)}+x_{\sigma(j)}}{x_{\sigma(i)}- x_{\sigma(j)}}.
\end{gather}
\end{Definition}
Alternatively, Schur $Q$-polynomial
$Q_\lambda =Q_\lambda (x_1,\dots, x_N)$ for $N>l$ is the coefficient of ${u^{-\lambda_1}\cdots u^{-\lambda_l}}$
in the formal generating function
\begin{gather}\label{genQ}
Q(u_1,\dots, u_l)=\sum_{\lambda_1,\dots, \lambda_l\in \bZ} \frac{Q_\lambda }{u^{\lambda_1}\cdots u^{\lambda_l}}=\prod_{1\le i<j\le l}\frac{u_j-u_i}{u_j+u_i} \prod_{i=1}^{l} Q(u_i),
\end{gather}
where it is understood that
\begin{gather*}
\frac{u_j-u_i}{u_j+u_i}= 1+2\sum_{r\ge 1}(-1)^r u_i ^{r} u_j^{-r},
\end{gather*}
and $Q(u)$ is given by (\ref{Qk1}) \cite[Chapter III, equation~(8.8)]{Md}.
Schur $Q$-polynomials have a stabilization property, hence, one can omit the number $N$ of variables $x_i's$ as long as it is not less than the length of the partition $\lambda$ and consider $Q_\lambda$ as functions of infinitely many variables $(x_1, x_2, \dots )$.

\section[Action of neutral fermions on bosonic space $\B_{\rm odd}$]{Action of neutral fermions on bosonic space $\boldsymbol{\B_{\rm odd}}$}\label{section3}

Consider the subalgebra $\B_{\rm odd}$ of $\B$ generated by odd ordinary Schur $Q$-functions: $\B_{\rm odd}= \bC[Q_1, Q_3,\dots]$. It is known that $\B_{\rm odd}$ is also a polynomial algebra in odd power sums $\B_{\rm odd}= \bC[p_1, p_3,\dots]$ and that Schur $Q$-functions $Q_\lambda$ labeled by strict partitions constitute a linear basis of $\B_{\rm odd}$ \cite[Chapter~III.8, equation~(8.9)]{Md}.

 Define operators $\{\varphi_k\}_{k\in \bZ}$ acting on the coefficients of generating functions $Q(u_1,\dots, u_l)$ by the rule
 \begin{gather}\label {phi}
 \Phi(v) Q(u_1,\dots, u_l)= Q(v, u_1,\dots, u_l)
 \end{gather}
with $\Phi(v)=\sum\limits_{m\in \bZ}{\varphi_m}{v^{-m}}$. Then in the expansion~(\ref{genQ})
\begin{gather*}
\varphi_m\colon \ Q_\lambda \mapsto Q_{(m,\lambda)}.
\end{gather*}
Observe that from (\ref{genQ})
 \begin{gather*}
 ( \Phi(u)\Phi(v) + \Phi(v)\Phi(u) )Q(u_1,\dots, u_l)\\
 \qquad{} =2\left(1+\sum_{r\ge 1}(-1)^r u ^{r} v^{-r} +\sum_{r\ge 1}(-1)^r v ^{r} u^{-r}\right) A(u,v, u_1,\dots, u_l)\\
\qquad{} =2\sum_{r\in\bZ} \left(\frac{-u}{v}\right)^r A(u,v, u_1,\dots, u_l),
 \end{gather*}
 where
 \begin{gather*}
 A(u,v, u_1,\dots, u_l)=\prod_{1\le j\le l}\frac{(u_j-u)(u_j-v)}{(u_j+u)(u_j+v)} Q(u) Q(v) Q(u_1,\dots, u_l).
 \end{gather*}
Using that $\delta(u,v)=\sum\limits_{r\in\bZ} {u^r}{v^{-(r+1)}}$ is a formal delta distribution with the property $ \delta(u,v) a(u)= \delta(u,v)a(v)$ for any formal distribution $a(u)=\sum\limits_{n\in\bZ} a_nu^n$, and that $Q(u)Q(-u)=1$, we get
 \begin{gather}\label{rel1}
 \left( \Phi(u)\Phi(v) + \Phi(v)\Phi(u) \right) Q(u_1,\dots, u_l)= 2v \delta(-u,v) Q(u_1,\dots, u_l).
 \end{gather}
 Since coefficients of the expansion of $Q(u_1,\dots, u_l)$ in powers of $u_1,\dots, u_l$ include Schur $Q$-functions $Q_\lambda$, and the latter form a linear basis of $\B_{\rm odd}$, it follows that (\ref{phi}) provides the action of well-defined operators $\{\varphi_k\}_{k\in \bZ}$
 on $\B_{\rm odd}$:
 \begin{gather*}
\varphi_{k} (Q_\lambda) = Q_{(k,\lambda)}.
\end{gather*}
 Relation (\ref{rel1}) on generating functions is equivalent to the commutation relations
 \begin{gather}\label{neut1}
[\varphi_m, \varphi_n]_+= 2(-1)^m \delta_{m+n,0}\qquad \text{for}\quad m, n\in \bZ.
\end{gather}
 Thus, operators $\{\varphi_i\}_{i\in \bZ}$ and ${1}$ provide the action of Clifford algebra ${\rm Cl}_\varphi$ of neutral fermions on the Fock space $\B_{\rm odd}$.
 Note that for any strict partition $\lambda=(\lambda_1>\lambda_2>\dots >\lambda_l)$
 \begin{gather}\label{ver1}
 Q_{(\lambda_1,\dots \lambda_l)}= \varphi_{\lambda_1}\cdots \varphi_{\lambda_l} ( 1),
 \end{gather}
 or in terms of generating functions,
 \begin{gather}\label{ver2}
 Q(u_1,\dots, u_l)= \Phi(u_1)\cdots \Phi(u_l) (1).
 \end{gather}
 Formulae (\ref{ver1}), (\ref{ver2}) sometimes are called the vertex operator realization of Schur $Q$-functions.

\section[Vertex operator form of formal distribution of neutral fermions]{Vertex operator form of formal distribution\\ of neutral fermions}\label{section4}
 It will be convenient for us to consider $\B_{\rm odd}$ as a subring of the ring of symmetric functions~$\B$. This allows us to recover the celebrated vertex operator form of the formal distribution of neutral fermions $\Phi(u)$ from no-less celebrated properties of generating functions of complete and elementary symmetric functions. All of these properties are discussed in \cite[Chapter~I]{Md}.

The ring of symmetric functions $\B$ possesses a bilinear form $(\cdot,\cdot)$ \cite[Chapter~I, equation~(4.5)]{Md} defined on the linear basis of monomials of power sums labeled by partitions~$\lambda$ and~$\mu$ as
\begin{gather*}(p_{\lambda_1}\cdots p_{\lambda_l},p_{\mu_1}\cdots p_{\mu_l})=z_{\lambda}\delta_{\lambda,\mu} ,\end{gather*}
 where
 $z_\lambda=\prod i^{m_i} m_i!$ and $m_i=m_i(\lambda)$ is the number of parts of~$\lambda$ equal to~$i$.

We will use this form and its restriction to $\B_{\rm odd}$ to define adjoint operators\footnote{Traditionally, one uses rescaled form on $\B_{\rm odd}$ defined as $(p_{\lambda}, p_{\mu})=2^{-l(\lambda)}z_{\lambda}\delta_{\lambda,\mu},$ where $l(\lambda)$ is the number of parts of $\lambda$, but rescaling is not necessary for our purposes, since in the rescaled form $p_n^\perp= n/2\cdot\partial /\partial{p_n}$ (see \cite[Chapter~III.8, Example~11]{Md}).}
of the multiplication operators.
By definition, given an element $f\in \B$, the operator $f^{\perp}$ adjoint to the operator of multiplication by $f$ is given by the rule
\begin{gather*}
\big(f^{\perp} g, h\big)= (g,fh)\qquad \text{for any} \quad g,h\in \B.
\end{gather*}
\cite[Chapter I.5, Example 3]{Md} contains the following statement. Consider a symmetric function $ f=f(p_1,p_2,\dots) $ expressed as a polynomial in power sums~$p_i$. Then the adjoint operator on $\B$ to the multiplication operator by~$f$ is given by
\begin{gather}\label{padj}
f^\perp= f\left(\frac{\partial}{\partial{p_1}},\frac{2\partial}{\partial{p_2}},\frac{3\partial}{\partial{p_3}}, \dots\right).
\end{gather}
In particular $p_n^\perp= n\partial /\partial{p_n}$.

Combine the corresponding adjoint operators of the families $h_k$, $e_k$, $p_k $ and $Q_k$ into generating functions
\begin{gather*}
H^{\perp}(u)=\sum_{k\ge 0} {h_k^{\perp}}{u^k},\qquad E^{\perp}(u)=\sum_{k\ge 0} {e_k^{\perp}}{u^k},\\ P^{\perp}(u)=\sum_{k\ge 1} {p_k^{\perp}}{u_k},\qquad
Q^{\perp}(u)=\sum_{k\ge 0} {Q_k ^{\perp}}{u_k}.
\end{gather*}
Then (\ref{padj}) immediately implies the following relations.
 \begin{Lemma}
\begin{gather*}
H^{\perp}(u)=\exp\left(\sum_{n\ge1}\frac{\partial}{\partial p_n} {u^{n}}\right),\qquad
E^{\perp}(u)=\exp\left(\sum_{n\ge1}{(-1)^{n-1}}\frac{\partial}{\partial p_n} {u^{n}}\right),
\\
Q^{\perp}(u) = E^{\perp}(u) H^{\perp}(u)= R^{\perp}(u)^2,
\qquad
R^{\perp}(u)=\exp\left(\sum_{n\in \bN_{\rm odd}} \frac{\partial}{\partial p_n}{u^{n}}\right),
\end{gather*}
where $\bN_{\rm odd}=\{1,3,5,\dots\}$.
\end{Lemma}
The proof of the next lemma is outlined in \cite[Chapter~I.5, Example~29]{Md}.
\begin{Lemma}The following commutation relations on generating functions of multiplication and adjoint operators acting on~$\B$ hold:
\begin{gather*}
H^{\perp}(u) \circ H(v) = (1-u/v)^{-1}H(u) \circ H^{\perp}(v),\\
H^{\perp}(u) \circ E(v) = (1+u/v)E(u) \circ H^{\perp}(v),\\
E^{\perp}(u) \circ H(v) = (1+u/v) H(u) \circ E^{\perp}(v),\\
E^{\perp}(u) \circ E(v) = (1-u/v)^{-1}E(u) \circ E^{\perp}(v).
\end{gather*}
\end{Lemma}

\begin{Corollary}\label{cor}
\begin{gather*}
H^{\perp}(u) \circ Q(v)= \frac{v+u}{v-u}Q(u) \circ H^{\perp}(v),\\
E^{\perp}(u) \circ Q(v)= \frac{v+u}{v-u}Q(u) \circ E^{\perp}(v),\\
R^{\perp}(u) \circ Q(v)= \frac{v+u}{v-u}Q(u) \circ R^{\perp}(v).
\end{gather*}
\end{Corollary}
\begin{proof}For the first and second one we use that $Q(u) = E(u) H(u)$. Observe that
\begin{gather*}
H^{\perp}(u)|_{\B_{\rm odd}}=E^{\perp}(u)|_{\B_{\rm odd}}=R^{\perp}(u)|_{\B_{\rm odd}}.
\end{gather*}
In other words, since $Q_k$ does not depend on even power sums~$p_{2r}$, we can add terms $ {\partial}/{\partial p_{2r}}$ in the sum under the exponent when applying to elements of~$\B_{\rm odd}$:
\begin{gather*}
R(u)^{\perp}(Q(v)) = \exp\left(\sum_{n\in \bN_{\rm odd}} \frac{\partial}{\partial p_n}{u^{n}}\right) Q(v)
= \exp\left(\sum_{n\ge 1} \frac{\partial}{\partial p_n}{u^{n}}\right) Q(v)= H^{\perp}(u) Q(v).\tag*{\qed}
\end{gather*}\renewcommand{\qed}{}
\end{proof}

We arrive at the vertex operator form of formal distribution of neutral fermions.
\begin{Proposition}
\begin{gather}\label{phiqr}
\Phi(v)= Q(v) R(-v)^{\perp}= \exp\left(\sum_{n\in \bN_{\rm odd}}\frac{2p_{n}}{n}\frac{1}{v^{n}}\right)\exp\left(-\sum_{n\in \bN_{\rm odd}} \frac{\partial}{\partial p_n}{v^{n}}\right).
\end{gather}
\end{Proposition}
\begin{proof}From Corollary \ref{cor}, the action of the operator $ Q(v) R(-v)^{\perp}$ on the coefficients of generating function $Q(u_1,\dots, u_l)$
coincides with the action of~$\Phi(v)$:
\begin{gather*}
 Q(v) R(-v)^{\perp} ( Q(u_1,\dots, u_l) ) = Q(v) \prod_{1\le i<j\le l}\frac{u_j-u_i}{u_j+u_i}\, R(-v)^{\perp}\left(
 \prod_{i=1}^{l} Q(u_i)\right)\\
\qquad{} = Q(v) \prod_{1\le j<i\le l}\frac{u_j-u_i}{u_j+u_i} \prod_{i=1}^{l}\frac{v-u_i}{v+u_i} \prod_{i=1}^{l} Q(u_i) =Q(v,u_1,\dots, u_l).
\end{gather*}
Since coefficients of $Q(u_1,\dots, u_l)$ contain a linear basis of $\B_{\rm odd}$, the equality~(\ref{phiqr}) follows.
\end{proof}

\section{The neutral fermions bilinear identity}\label{section5}

Let
\begin{gather*}
\Omega= \sum_{n}\varphi_n\otimes (-1)^n \varphi_{-n}.
\end{gather*}
One looks for the solutions in $\B_{\rm odd}$ of {\it the neutral fermions bilinear identity}
\begin{gather}\label{binBKP}
\Omega ( \tau\otimes \tau)= \tau\otimes\tau,
\end{gather}
where $\tau= \tau(\tilde p)= \tau ( 2p_1, 2p_3/3, 2p_5/5,\dots )$.
It is known \cite{DKTIV,DKTA,DKTII,JM} that (\ref{binBKP}) is equivalent to an infinite integrable system of partial differential equations called the BKP hierarchy.
Further in Section~\ref{section6} a simple observation explains, why Schur $Q$-functions constitute solutions of the neutral fermions bilinear identity, and hence of the BKP hierarchy.

In this section we would like to make a small deviation and review the steps of recovering the BKP hierarchy of partial differential equations in the Hirota form from the neutral fermions bilinear identity. This is certainly a well-known procedure. However, the explicit calculations are often omitted in the literature, and we would like to provide them here for the convenience of the reader.

Note that $\Omega$ is the constant coefficient of the formal distribution $\Phi(u)\otimes \Phi(-u)$, or, in terms of residue,
\begin{gather}\label{resid}
\Omega=\mathop{\operatorname{Res}}\limits_{u=0} \frac{1}{u} \Phi(u)\otimes \Phi(-u).
\end{gather}
We identify $\B_{\rm odd}\otimes \B_{\rm odd}$ with $\bC[p_1,p_3,\dots]\otimes \bC[r_1,r_3,\dots]$ -- two copies of polynomial rings, where variables in each of them play role of power sum symmetric functions. Set $ \tilde p= ( 2p_1, 2p_3/3,\allowbreak 2p_5/5,\dots )$, $ \tilde r=(2r_1,2r_3/3, 2r_5/5,\dots)$.
Then $\partial_{p_n}=2\partial_{\tilde p_n}/n$ and
\begin{gather*}
\Phi(u)\tau\otimes \Phi(-u) \tau = \exp\left(\sum_{n\in \bN_{\rm odd}}(\tilde p_{n}- \tilde r_n)\frac{1}{u^{n}}\right)\\
\hphantom{\Phi(u)\tau\otimes \Phi(-u) \tau =}{}\times \exp\left(-\sum_{n\in \bN_{\rm odd}}\frac{2}{n}\left(\frac{\partial}{\partial \tilde p_n}-\frac{\partial}{\partial \tilde r_n}\right){u^{n}}\right)\tau( \tilde p)\tau(\tilde r).
\end{gather*}
Introduce the change of variables
\begin{gather*}
\tilde p_n=x_n-y_n,\qquad \tilde r_n=x_n+y_n.
\end{gather*}
Then
\begin{gather*}
\Phi(u)\tau\otimes \Phi(-u) \tau
= \exp\left(\sum_{n\in \bN_{\rm odd}} -2y_n \frac{1}{u^n}\right)\exp\left(\sum_{n\in \bN_{\rm odd}}\frac{2} {n}\frac{\partial}{\partial y_n}u^{n}\right)\tau( x- y)\tau( x+ y)
\end{gather*}
with $(x\pm y)=(x_1\pm y_1, x_3\pm y_3, x_5\pm y_5,\dots)$.

\begin{Definition} Let $P(D )$ be a multivariable polynomial in the collection of variables $D=(D_1, D_2,\dots)$, let $f(x)$, $g(x)$ be differentiable functions in $ x =(x_1,x_2,\dots)$.

The {\it Hirota derivative} $P( D) f \cdot g$ is a function in variables $(x_1,x_2,\dots)$ given by the expression
\begin{gather*}
P( D) f \cdot g =P(\partial_{z_1},\partial_{z_1},\dots)f (x+ z)g( x- z)|_{ z=0},
\end{gather*}
where $ x\pm z =(x_1\pm z_1,x_2\pm z_2,\dots)$.
\end{Definition}
For example,
\begin{gather*}
 D_i^n f \cdot g=\sum_{k=0}^{n} (-1)^k {n\choose k} \frac{\partial ^k f}{\partial x_i^k}\frac{\partial ^{n-k} g}{\partial x_i^{n-k}},
\end{gather*}
which implies in particular that odd Hirota derivatives are tautologically zero when $f=g$:
\begin{gather*}
 D_i^{2n+1} f \cdot f =0\qquad \text{for} \quad n=0,1,2,\dots.
\end{gather*}
The following lemma allows one to rewrite bilinear identity (\ref{binBKP}) in terms of the Hirota derivatives.
\begin{Lemma}
\begin{gather*}
\exp\left(\sum_{n\in \bN_{\rm odd}}\frac{2}{n}\frac{\partial}{\partial y_n}{u^{n}}\right)\tau( x- y)\tau(x+y)
=\exp\left(\sum_{n\in \bN_{\rm odd}}\left(y_n+\frac{2}{n}u^n\right)D_n\right)\tau \cdot\tau.
\end{gather*}
\end{Lemma}
\begin{proof}By the Taylor series expansion,
\begin{gather}\label{taylor}
{\rm e}^{a\partial/\partial_y} g(y)=\sum_{n=0}^{\infty} \frac{ a^ng^{(n)} (y)} {n!}= g(y+a).
\end{gather}
Applying (\ref{taylor}) twice with $ t=(t_1,t_3, t_5,\dots)$, $ \tilde u=\big(2u, 2u^3/3, 2u^5/5,\dots \big)$,
\begin{gather*}
\exp\left(\sum_{n\in \bN_{\rm odd}}\frac{2}{n}\frac{\partial}{\partial y_n} u^n\right) \tau( x- y)\tau( x+y)= \tau( x+ y+ \tilde u)\tau( x- y-\tilde u)\\
\qquad{} =\tau(x+ y +\tilde u+ t)\tau(x-( y+\tilde u+t))|_{ t=0}\\
\qquad{} =\left.\exp\left(\sum_{n\in \bN_{\rm odd}}\left(y_n+\frac{2}{n}u^n\right)\frac{\partial}{\partial t_n }\right)\tau( x+ t)\tau( x- t)\right\vert_{ t=0}.\tag*{\qed}
\end{gather*}\renewcommand{\qed}{}
\end{proof}

Thus, we can write in terms of Hirota derivatives
\begin{gather}
\Phi(u)\tau\otimes \Phi(-u) \tau =\exp\left(\sum_{n\in \bN_{\rm odd}} \frac{-2y_n}{u^n}\right)\nonumber\\
\hphantom{\Phi(u)\tau\otimes \Phi(-u) \tau =}{}\times \exp\left(\sum_{n\in \bN_{\rm odd}}\frac{2}{n}D_n u^n\right) \exp\left(\sum_{n\in \bN_{\rm odd}}y_nD_n\right) \tau\cdot\tau.\label{eq1}
\end{gather}

In order to compute $\mathop{\operatorname{Res}}\limits_{u=0} \frac{1}{u} \Phi(u)\tau\otimes \Phi(-u)\tau$, which is just the coefficient of~$u^0$ of $ \Phi(u)\tau\otimes \Phi(-u) \tau $, we recall the following well-known facts on the composition of exponential series with generating series. Their proofs can be done, e.g., by induction, or again found in~\cite[Chapter~I]{Md}.
\begin{Proposition}\label{XY} Let $S(u)=\sum\limits_{k=0}^\infty S_k \frac{1}{u^k}$ and $X(u)=\sum\limits_{k=1}^\infty X_k \frac{1}{u^k}$ be related by
\begin{gather*}
\exp ( X(u)) =S(u).
\end{gather*}
Then the following statements hold
\begin{gather*}
S_k=\sum_{s=1}^{k}\sum_{l_1+2l_2+\dots +sl_s=k,\, l_i\ge 1}\frac{1}{l_1!\cdots l_s!}X_1^{l_1}\cdots X_s^{l_s},\\
S_k=\det\frac{1}{n!}
\begin{pmatrix}
X_1 &-1&0&0&\dots&0\\
2X_2 &X_1&-2&0&\dots&0\\
3X_3 &2X_2&X_1&-3&\dots&0\\
\dots&\dots&\dots&\dots&\dots&0\\
kX_k &(k-1)X_{k-1}&(k-2)X_{k-2}&(k-3)X_{k-3}&\dots&X_1\\
\end{pmatrix},\\
X_k=\frac{(-1)^{k-1}}{k} \det
\begin{pmatrix}
S_1 &1&0&0&\dots&0\\
2S_2 &S_1&1&0&\dots&0\\
3S_3 &S_2&S_1&1&\dots&0\\
\dots&\dots&\dots&\dots&\dots&0\\
kS_k &S_{k-1}&S_{k-2}&S_{k-3}&\dots&S_1\\
\end{pmatrix} .
\end{gather*}
\end{Proposition}
\begin{Example}
\begin{gather*}
S_0=1,\qquad S_1= X_1,\qquad S_2= \frac{1}{2} X_1^2+X_2,\qquad S_3= S_3+X_2X_1 +\frac{1}{6}X_1^3,\\
S_4=X_4+X_3X_1 +\frac{1}{2}X_2^2+\frac{1}{2} X_2X_1^2+\frac{1}{24}X_1^4.
\end{gather*}
By Lemma \ref{Lemma1}, when $X$ variables in these formulae are interpreted as normalized power sums $X_{k}=p_{k}/k$, $S_k$'s are identified with complete symmetric functions~$h_k$'s.
\end{Example}
\begin{Example}\label{eg2} Let $X_{2k}=0$ for $k=1,2,\dots.$ Then the first few $S_n=S_n(X_1, 0, X_3,\dots)$ are given by
\begin{gather*}
S_0 =1,\qquad S_1= X_1,\qquad S_2= \frac{1}{2}X_1^2,\qquad S_3= X_3+\frac{1}{6}X_1^3,\qquad
S_4 =X_3X_1 +\frac{1}{24}X_1^4,\\ S_5= \frac{1}{120}X_1^5+\frac{1}{2}X_1^2X_3+X_5,\qquad
S_6 =\frac{1}{720} X_1^6+\frac{1}{6}X_1^3X_3+\frac{1}{2}X_3^2 +X_1X_5,\\
S_7 =\frac{1}{5040} X_1^7+\frac{1}{24}X_1^4X_3+\frac{1}{2} X_1X_3^2 +\frac{1}{2}X_1X_5 +X_7.
\end{gather*}
Note that by Lemma \ref{Lemma2} when $X$ variables in these formulae are interpreted as odd normalized power sums $X_{2k+1}=2p_{2k+1}/(2k+1)$, $S_k$'s are identified with Schur $Q$-functions $Q_k$'s.
\end{Example}

Using the statement of Proposition \ref{XY}, we can write the coefficient of $u^0$ of (\ref{eq1}) as
\begin{gather}\label{BKP3}
\sum_{m=0}^{\infty}
S_{m}\left( \tilde y \right)
S_{m}\big( \tilde D\big)
\exp\left(\sum_{n\in \bN_{\rm odd}}y_nD_n\right)
\tau \cdot\tau =\tau(x-y)\cdot\tau(x+y),
\end{gather}
where $\tilde y= ( {-2y_1},0, -2y_3, \dots )$, $\tilde D=( 2D_1, 0,2D_3/3, 0,\dots)$.

Note that $S_0=1$ and $\exp\Big(\sum\limits_{n\in \bN_{\rm odd}}y_nD_n\Big) \tau \cdot\tau =\tau(x- y)\cdot\tau( x+ y)$, hence we can rewrite~(\ref{BKP3}) as
\begin{gather}\label{BKP2}
\sum_{m=1}^{\infty} S_{m}(\tilde{y})S_{m}\big( \tilde D\big)\exp\left(\sum_{n\in \bN_{\rm odd}}y_nD_n\right) \tau \cdot\tau=0.
\end{gather}
To obtain the equations of the BKP hierarchy, one expands the left hand side of (\ref{BKP2}) in monomials $y_1^{m_1}y_2^{m_2}\cdots y_N^{m_N}$ to obtain as coefficients Hirota operators in terms of $D_k$'s.

For example, let us compute the coefficient of $y_3^2$. In the expansion of
 \begin{gather*}
\big(S_{1}(\tilde{y})
S_{1}\big( \tilde D\big) + S_{2}(\tilde{y})
S_{2}\big( \tilde D\big)+ S_{3}(\tilde{y})
S_{3}\big( \tilde D\big) +\cdots\big) \left(1+\sum y_i D_i+ \frac{1}{2}\left(\sum y_i D_i\right)^2+\cdots \right)
 \end{gather*}
the term $y_3^2$ appears in $S_{3}(\tilde{y}) S_{3}\big( \tilde D\big)\times y_3D_3$ and in $S_{6}(\tilde{y}) S_{6}\big( \tilde D\big)\times 1$. Using the expansions of Example~\ref{eg2}, the coefficient of $y_3^2$ is
\begin{gather*}
-2S_{3}\big( \tilde D\big)D_3+ 2S_{6}\big( \tilde D\big)= \frac{8}{45}\big(D_1^6- 5D_1D_3 -5 D_3^2+9D_1D_5\big),
\end{gather*}
which provides the Hirota bilinear form of the BKP equation that gives the name to the hierarchy
\begin{gather*}
\big(D_1^6- 5D_1D_3 -5 D_3^2+9D_1D_5\big)\tau\cdot \tau=0.
 \end{gather*}

\begin{Remark}Writing the residue (\ref{resid}) as a contour integral, one gets the BKP in its {\it integral form}
\begin{gather*}
\oint \frac{1}{2\pi {\rm i} u} \exp\left(\sum_{n\in \bN_{\rm odd}}(\tilde p_n-\tilde r_n)\frac{1}{u^n}\right)\exp\left(-\sum_{n\in \bN_{\rm odd}}\frac{2}{n}\left(\frac{\partial}{\partial \tilde p_n}-\frac{\partial}{\partial \tilde r_n}\right){u^{n}}\right)\tau( \tilde p)\tau(\tilde r)\\
\qquad{} =\tau( \tilde p)\tau(\tilde r).
\end{gather*}
\end{Remark}

\section{Commutation relation for the bilinear identity}\label{section6}
Our goal is to show that multiparameter Schur $Q$-functions are solutions of the neutral fermions bilinear identity (\ref{binBKP}), thus they provide solutions of the BKP hierarchy.

Let $X=\sum\limits_{n> 0} A_n\varphi_n$ for some $A_n\in \bC$. From~(\ref{neut1}) one gets $X^2=0$.
\begin{Proposition}
	\begin{gather*}
		\Omega (X\otimes X)= (X\otimes X)\Omega.
	\end{gather*}
\end{Proposition}
\begin{proof}
\begin{align*}
\Omega (X\otimes X)&=\sum_{k\in \bZ} \varphi_k X\otimes (-1)^{k}\varphi_{-k} X\\
&=\sum_{k\in \bZ} (-X\varphi_k+[\varphi_k, X]_+)\otimes (-1)^{k}(-X\varphi_{-k} +[\varphi_{-k}, X]_+).
\end{align*}
Note that $[\varphi_{k}, X]_+=\big[\varphi_{k}, \sum\limits_{n>0}A_n\varphi_n\big]_+=2\sum\limits_{n>0}(-1)^n A_n\delta_{k+n,0}$, hence
\begin{gather*}
\Omega (X\otimes X) = (X\otimes X)\Omega - \sum_{n>0}2(-1)^n A_n\otimes (-1)^n X\varphi_n - \sum_{n>0}(-1)^n X \varphi_n \otimes 2(-1)^n A_n\\
\hphantom{\Omega (X\otimes X) =}{} +4\sum_{k\in \bZ}\sum_{m,n>0}(-1)^n A_n\delta_{n+k,0}\otimes (-1)^m A_m\delta_{m-k,0}\\
\hphantom{\Omega (X\otimes X)}{} = (X\otimes X)\Omega - 2\otimes X^2- X^2 \otimes 2+4\sum_{m,n>0}(-1)^n A_n\otimes (-1)^m A_m\delta_{m+n,0}.
\end{gather*}
We use that $X^2= 0$, and since both $m$ and $n$ in the last sum are always positive, the last term is also zero.
 \end{proof}

\begin{Corollary}\label{cor8} Let $\tau\in \B_{\rm odd}$ be a solution of~\eqref{binBKP}, and let $X=\sum\limits_{n> 0} A_{n} \varphi_n$ with $A_{n}\in \bC$. Then $\tau^\prime =X\tau$ is also a solution of~\eqref{binBKP}.
 \end{Corollary}
The vertex operator presentation~(\ref{ver1}) of Schur $Q$-functions and Corollary~\ref{cor8} immediately imply that Schur $Q$-functions are solutions of~(\ref {binBKP}), since constant function~$1$ is a~solution of~(\ref{binBKP}). This argument reproves the result of~\cite{You1} and easily extends to more general case of multiparameter Schur $Q$-functions defined in the next section.

\section[Multiparameter Schur $Q$-functions are solutions of the BKP hierarchy]{Multiparameter Schur $\boldsymbol{Q}$-functions are solutions\\ of the BKP hierarchy}\label{multi_sec}

Multiparameter Schur $Q$-functions were introduced in~\cite{Iv1} as a generalization of definition (\ref{defQ}). Fix an infinite sequence of complex numbers $a= (a_0, a_1, a_2,\dots)$. Consider the analogue of a~power of a~variable~$x$
\begin{gather*}
(x|a)^{k}= (x-a_0)(x-a_2)\cdots (x- a_{k-1}).
\end{gather*}
We also define a shift operation $\tau\colon a_k\mapsto a_{k+1}$, so that
\begin{gather*}
(x|\tau a)^{k}= (x-a_1)(x-a_2)\cdots (x- a_{k}).
\end{gather*}

\begin{Definition} Let $\alpha=(\alpha_1,\dots,\alpha_l)$ be a vector with non-negative integer coefficients $\alpha_i\in \bZ_{\ge 0}$. The multiparameter Schur $Q$-function in variables $(x_1,\dots, x_N)$ with $l\le N$ is defined by
\begin{gather}\label{defQa}
Q^{(a)}_\alpha (x_1,\dots, x_N)=\frac{2^l}{(N-l)!} \sum_{\sigma\in S_N} \prod_{i=1}^{l}(x_{\sigma(i)}|a)^{\alpha_i} \prod_{i\le l,i<j\le N }\frac{x_{\sigma(i)}+x_{\sigma(j)}}{x_{\sigma(i)}- x_{\sigma(j)}}.
\end{gather}
\end{Definition}

When $a=(0,0,0,\dots )$ and $\alpha$ is a strict partition, one gets back (\ref{defQ}), the classical Schur $Q$-functions $Q_\alpha (x_1,\dots, x_N)$. The evaluation $a=(0,1,2,\dots)$ gives factorial Schur $Q$-functions denoted as $Q^*_\alpha(x)$, those applications are outlined in the introduction. The multiparameter Schur $Q$-functions enjoy a~stability property, hence one can consider $Q^{(a)}_\alpha(x_1,x_2,\dots)$ to be a~function of infinitely many variables.

Note from (\ref{defQa}) that for any permutation $\sigma \in S_l$,
\begin{gather}\label{alter}
Q^{(a)}_{\alpha} (x_1,\dots, x_N)= (-1)^\sigma Q^{(a)}_{\sigma(\alpha)} (x_1,\dots, x_N),
\end{gather}
where $(-1)^\sigma$ is the sign of permutation $\sigma$ \cite[Proposition~3]{Kor}. Hence, $Q^{(a)}_{\alpha} =0$ if $\alpha_i=\alpha_j$ for some~$i$,~$j$, and for a vector $\alpha=(\alpha_1,\dots,\alpha_l)$ with positive distinct integer coefficients $\alpha_i\in \bZ_{>0}$, function $Q^{(a)}_{\alpha}$ coincides up to a sign with another $Q^{(a)}_{\alpha^\prime}$ labeled by strict partition~$\alpha^\prime$.

One can check directly the following transitions between regular and multiparameter powers of variables.
\begin{Lemma}\label{transitions}For $n=0,1,2,\dots$
\begin{gather*}
(u-a_1)\cdots (u-a_{n}) =\sum_{k=0}^{\infty} (-1)^{n-k}e_{n-k}(a_1,\dots, a_{n}) u^k, \\
\frac{1}{(u-a_1)\cdots (u-a_{n})} =\sum_{k=0}^{\infty}h_{k-n}(a_1,\dots, a_{n}) u^{-k}, \\
u^n =\sum_{k=0}^{\infty} h_{n-k}(a_1,\dots, a_{k+1}) (u-a_1)\cdots (u-a_k), \\
\frac{1}{u^{n}} =\sum_{k=0}^{\infty} (-1)^{n-k}e_{k-n}(a_1,\dots, a_{k-1}) \frac{1}{(u-a_1)\cdots (u-a_k)}.
\end{gather*}
\end{Lemma}

Double application of Lemma \ref{transitions} implies the following useful observation.
\begin{Lemma}\label{identity} For any sequence $a=(0, a_1, a_2,\dots)$
\begin{gather*}
\sum_{m=0}^{\infty} \frac{(x|a)^m}{(u|\tau a)^m}= \sum_{m=0}^{\infty} \frac{x^m}{u^m}.
\end{gather*}
\end{Lemma}
\begin{proof}
\begin{align*}
\sum_{m=0}^{\infty} \frac{(x|a)^m}{(u|\tau a)^m}& = \sum_{m,k=0}^{\infty} (-1)^{m-k}e_{m-k}(0, a_1\cdots a_{m-1}) x^k\frac{1}{(u|\tau a)^m}\\
& =\sum_{k=0}^{\infty} x^k\sum_{m=0}^{\infty} (-1)^{k-m}e_{m-k}(a_1\cdots a_{m-1}) \frac{1}{(u|\tau a)^m}= \sum_{k=0}^{\infty} \frac{x^k}{u^k}.\tag*{\qed}
\end{align*}\renewcommand{\qed}{}
\end{proof}

Consider a part of the generating function~(\ref{genQ}) of ordinary Schur $Q$-functions that corresponds only to positive values of $\lambda_i$:
\begin{gather*}
Q^+(u_1,\dots, u_l)= \sum_{\lambda_1,\dots, \lambda_l\in \bZ_{>0}} \frac{Q_\lambda}{u_1^{\lambda_1}\cdots u_l^{\lambda_l}}.
\end{gather*}
By (\ref{alter}), every non-zero coefficient of $Q^+(u_1,\dots, u_l)$ up to a sign coincides with a classical Schur $Q$-function labeled by an appropriate strict partition. In~\cite{Kor} the following remarkable observation is made.
\begin{Theorem}[\cite{Kor}]\label{Korotkih} For any sequence $a=(0, a_1, a_2,\dots)$
\begin{gather*}
Q^+(u_1,\dots, u_l)= \sum_{\lambda_1,\dots, \lambda_l\in \bZ_{> 0}} \frac{Q^{(a)}_\lambda}{(u_1|\tau a)^{\lambda_1}\cdots (u_l|\tau a)^{\lambda_l}}.
\end{gather*}
\end{Theorem}
\begin{proof} In~\cite{Kor} theorem is proved by induction on the length of the vector $\lambda$. A very short proof of this theorem follows from Lemma~\ref{identity} and definition (\ref{defQa}). Indeed,
 \begin{gather*}
 \sum_{\lambda_i \in \bZ_{>0}} \frac{Q^{(a)}_\lambda}{(u_1|\tau a)^{\lambda_1}\cdots (u_l|\tau a)^{\lambda_l}}
 =\frac{2^l}{(N-l)!}\sum_{\sigma\in S_N} \prod_{i=1}^{l}\sum_{\lambda_i\in \bZ_{> 0}}\frac{(x_{\sigma(i)}|a)^{\lambda_i}}{(u_i|\tau a)^{\lambda_i}} \prod_{i\le l,i<j\le N }\frac{x_{\sigma(i)}+x_{\sigma(j)}}{x_{\sigma(i)}- x_{\sigma(j)}}\\
 \qquad{} =\frac{2^l}{(N-l)!}\sum_{\sigma\in S_N} \prod_{i=1}^{l}\sum_{\lambda_i\in \bZ_{> 0}}\frac{x_{\sigma(i)}^{\lambda_i}}{u_i^{\lambda_i}} \prod_{i\le l,i<j\le N }\frac{x_{\sigma(i)}+x_{\sigma(j)}}{x_{\sigma(i)}- x_{\sigma(j)}}=Q^+(u_1,\dots, u_l).\tag*{\qed}
 \end{gather*}\renewcommand{\qed}{}
\end{proof}

Thus, Theorem \ref{Korotkih} suggests that multiparameter Schur $Q$-functions are obtained from classical Schur $Q$-functions by the change of the basis of expansion $\big\{1/u^{k}\big\} \to\big\{ 1/(u|\tau a)^k\big\}$ in the generating function $Q^+(u_1,\dots, u_l)$.

Lemma \ref{transitions} immediately implies the following relations between classical and multiparameter Schur $Q$-functions (see also \cite[Theorem~10.2]{Iv1})
\begin{Corollary} \label{cor10} For any sequence of complex numbers $a= (0, a_1, a_2,\dots)$ and any integer vector $\alpha =(\alpha_1,\dots, \alpha_l)$ with $\alpha_i\in \bZ_{>0}$,
\begin{gather*}
Q^{(a)}_{\alpha}=\sum_{\lambda_1,\dots, \lambda_l \in \bZ_{> 0}} (-1)^{\sum \lambda_i-\sum\alpha_i}e_{\alpha_1-\lambda_1}(a_1,\dots, a_{\alpha_1-1})
\cdots e_{\alpha_l-\lambda_l}(a_1,\dots, a_{\alpha_l-1}) {Q_\lambda}.
\end{gather*}
\end{Corollary}

\begin{Theorem} For any sequence of complex numbers $a= (0, a_1, a_2,\dots)$ and any integer vector $\alpha =(\alpha_1,\dots, \alpha_l)$ with $\alpha_i\in \bZ_{> 0}$,
 multiparameter Schur $Q$-function $Q^{(a)}_\alpha$ is a solution of~\eqref{binBKP}.
 \end{Theorem}
 \begin{proof}The constant polynomial $1$ is obviously a solution of~(\ref{binBKP}) in $\B_{\rm odd}$. By the vertex operator presentation~(\ref{ver1}) and Corollary~\ref{cor10},
\begin{gather*}
Q^{(a)}_{\alpha}=
\sum_{\lambda_1,\dots, \lambda_l\in \bZ_{>0}} A_{\lambda_1,\alpha_1}\cdots A_{\lambda_l,\alpha_l} \varphi_{\lambda_l}\cdots \varphi_{\lambda_1} \cdot 1
\end{gather*}
 with $A_{n,k}= (-1)^{n-k}e_{k-n}(a_1,\dots, a_{k-1})$. Hence,
\begin{gather*}
Q^{(a)}_{\alpha}= X_{\alpha_l}\cdot X_{\alpha_1} \cdot 1,
\end{gather*}
where
$ X_{m} =\sum\limits_{s>0} (-1)^{m-s}e_{s-m}(a_1,\dots, a_{s-1})\varphi_{s}$, and Corollary~\ref{cor8} proves the statement.
\end{proof}

\subsection*{Acknowledgements}
The author would like to thank the referee for the thoughtful and careful review that helped to improve the text of the paper.

\pdfbookmark[1]{References}{ref}
\LastPageEnding

\end{document}